\documentclass[preprint,12pt]{elsarticle}

\usepackage{amssymb}
\usepackage{amsmath}
\usepackage{amsthm}
\usepackage{mathtools}
\usepackage{algorithm}
\usepackage{algpseudocode}
\usepackage{bm}

\newtheorem{lemma}{Lemma}[section]
\newtheorem{theorem}{Theorem}[section]
\newtheorem{corollary}{Corollary}[section]
\newtheorem{proposition}{Proposition}[section]
\newtheorem{remark}{Remark}[section]

\newtheorem{observation}{Observation}[section]
\algrenewcommand\algorithmicrequire{\textbf{Input:}}
\algrenewcommand\algorithmicensure{\textbf{Output:}}

\begin{document}

\begin{frontmatter}

\title{Removing bottlenecks in the recognition of small $(k,\ell)$-graph classes} 

\author[1,2]{Flavia Bonomo-Braberman}

\author[1,3]{Min Chih Lin}

\author[1,2]{Ignacio~Maqueda}

\affiliation[1]{organization={Universidad de Buenos Aires, Facultad de Ciencias Exactas y Naturales. Departamento de Computación},
            addressline={Pres. Dr. Raúl Alfonsín s/n},
             city={Buenos Aires},
             postcode={1428},
             country={Argentina}}

\affiliation[2]{organization={CONICET-Universidad de Buenos Aires, Instituto de Investigación en Ciencias de la Computación (ICC)},
            addressline={Pres. Dr. Raúl Alfonsín s/n},
             city={Buenos Aires},
             postcode={1428},
             country={Argentina}}

\affiliation[3]{organization={CONICET-Universidad de Buenos Aires, Instituto de Cálculo (IC)},
            addressline={Pres. Dr. Raúl Alfonsín s/n},
             city={Buenos Aires},
             postcode={1428},
             country={Argentina}}

%% Abstract
\begin{abstract}
A graph is a $(k,\ell)$-graph if its vertex set can be partitioned into
$k$ independent sets and $\ell$ cliques. This family simultaneously
generalizes split, bipartite, and co-bipartite graphs. While the recognition
problem is NP-complete whenever $k\geq 3$ or $\ell\geq 3$, the remaining
small cases are polynomial-time solvable. In this paper we revisit the
known recognition algorithms for the first nontrivial polynomial cases,
namely $(2,1)$-, $(1,2)$-, and $(2,2)$-graphs, and show how to remove
specific bottlenecks in their existing recognition procedures. For
$(2,1)$-graphs, we show that the extra quadratic enumeration in the
algorithm of Brandst\"adt, Le and Szymczak can be avoided by exploiting the
structure of a shortest odd cycle in the relevant residual graph, reducing
the running time from $O((n+m)^2)$ to $O(n(n+m))$. By complementation, this
yields an $O(n(n+\overline m))$-time recognition algorithm for
$(1,2)$-graphs, where $\overline m$ denotes the number of edges of the
complement graph. For $(2,2)$-graphs, we refine the sparse--dense partition
framework of Feder, Hell, Klein and Motwani by restricting the local-search
enumeration to sets that are simultaneously bipartite and co-bipartite, and
by using the improved algorithms for $(2,1)$- and $(1,2)$-graphs as
preprocessing tools. This gives an
$O(n^4(n+\min\{m,\overline m\})^3)$-time recognition algorithm for
$(2,2)$-graphs.
\end{abstract}

%% Keywords
\begin{keyword}
$(k,\ell)$-graphs \sep graph partitions \sep recognition algorithms \sep sparse--dense partitions \sep bipartite and co-bipartite graphs

\end{keyword}

\end{frontmatter}

\section{Introduction}

A graph $G=(V,E)$ is called a \emph{$(k,\ell)$-graph} if its vertex set $V$
can be partitioned into $k$ independent sets and $\ell$ cliques. This notion
was introduced by Brandst\"adt~\cite{BRANDSTADT199647}. Independently,
Alekseev and Lozin~\cite{1116-AlekseevL03} later considered the same type
of mixed partitions under the name of \emph{$(p,q)$-colorable graphs}, where
the vertices are partitioned into at most $p$ cliques and $q$ independent
sets. They proved that the maximum-weight independent set problem is
polynomial-time solvable in this class if and only if $q \leq 2$, unless
$P=NP$.

The class of $(k,\ell)$-graphs provides a common framework for several
classical graph classes. Split graphs are precisely the $(1,1)$-graphs,
bipartite graphs are the $(2,0)$-graphs, and co-bipartite graphs are the
$(0,2)$-graphs. These three classes admit linear-time recognition algorithms.
In contrast, the recognition problem for $(k,\ell)$-graphs is NP-complete
whenever $k\geq 3$ or $\ell\geq 3$~\cite{BRANDSTADT199647,BRANDSTADT1998295}.
Thus, after the linear-time cases with $k+\ell\leq 2$, the first nontrivial
polynomial cases are $(2,1)$-, $(1,2)$-, and $(2,2)$-graphs.

The algorithmic history of these boundary cases is somewhat delicate.
Brandst\"adt~\cite{BRANDSTADT199647} originally proposed polynomial-time
recognition algorithms for $(2,1)$-, $(1,2)$-, and $(2,2)$-graphs, but a
subsequent corrigendum~\cite{BRANDSTADT1998295} reported that those
algorithms were incorrect. Brandst\"adt, Le and Szymczak~\cite{BRANDSTADT199859}
then gave an $O((|V|+|E|)^2)$-time recognition algorithm for $(2,1)$-graphs.
By complementation, this also yields an $O((n+\overline m)^2)$-time
algorithm for $(1,2)$-graphs, where $n=|V|$ and
$\overline m=\binom n2-m$ is the number of edges of the complement graph.
Feder, Hell, Klein and Motwani~\cite{DBLP:conf/stoc/FederHKM99}
studied a broader family of sparse--dense partition problems. Their
framework gives, in particular, an
$O(|V|^{10}(|V|+|E|))$-time recognition algorithm for $(2,2)$-graphs.

In this paper we revisit these recognition algorithms and show how to remove
specific bottlenecks in their existing recognition procedures. The improvements
come from exploiting additional structure in the instances left unresolved by
the previous algorithms. For $(2,1)$-graphs, the only part of the algorithm of
Brandst\"adt, Le and Szymczak that exceeds the $O(n(n+m))$ bound is the
residual case in which a certain undecided set induces no triangle. We show
that, in this case, the extra quadratic enumeration can be avoided: it is
enough to inspect a linear number of candidate vertices or pairs, obtained
from the structure of a shortest odd cycle in the residual graph. This reduces
the running time from $O((n+m)^2)$ to $O(n(n+m))$.

The case of $(1,2)$-graphs follows by applying the same argument to the
complement graph. Hence the running time becomes
$O(n(n+\overline m))$, where $\overline m=\binom n2-m$ denotes the number
of edges of the complement.

For $(2,2)$-graphs we refine the sparse--dense framework of
Feder et al.~\cite{DBLP:conf/stoc/FederHKM99}. The key point is a general
feature of the sparse--dense local search: when the current sparse side is
compared with a target sparse side, the set removed from the current sparse
side is contained in the intersection of the current sparse side and the target
dense side. In the bipartite/co-bipartite setting, where both classes are
hereditary, this means that the removed set induces a graph that is
simultaneously bipartite and co-bipartite, which substantially restricts the
relevant choices. We combine this observation with a preprocessing step based
on the improved recognition algorithms for $(2,1)$- and $(1,2)$-graphs. The
resulting algorithm recognizes $(2,2)$-graphs in time
$O(n^4 (n+\min\{m,\overline m\})^3)$.

Table~\ref{tab:table2} summarizes the previously best known bounds and the
improvements obtained in this paper.

\begin{table}[h!]
  \begin{center}
    \begin{tabular}{c|c|c|c}
      \textbf{Graph Class} & \textbf{Previous bound} & \textbf{Ref.} & \textbf{This paper} \\
      \hline
      $(2,1)$-graphs & $O((n+m)^2)$ & \cite{BRANDSTADT199859} & $O(n(n+m))$ \\
      $(1,2)$-graphs & $O((n+\overline m)^2)$ & \cite{BRANDSTADT199859} & $O(n(n+\overline m))$ \\
      $(2,2)$-graphs & $O(n^{10}(n+m))$ & \cite{DBLP:conf/stoc/FederHKM99} & $O(n^4 (n+\min\{m,\overline m\})^3)$ \\
    \end{tabular}
    \caption{Complexity comparison between the previously best known recognition algorithms and the algorithms presented in this paper.}
    \label{tab:table2}
  \end{center}
\end{table}

Related work has also considered $(k,\ell)$-graph recognition within
specific graph classes, including chordal graphs~\cite{DBLP:journals/dam/HellKNP04},
cographs~\cite{DBLP:journals/disopt/DemangeEW05}, and line graphs~\cite{DBLP:journals/tcs/DemangeEW05}.
Several probe variants have also been studied, both in the partitioned
setting~\cite{DANTAS201867} and in the unpartitioned
setting~\cite{DANTAS2016294}.

\section{Preliminaries}\label{sec:Preliminaries}

We use standard graph-theoretic notation. All graphs considered in this
paper are finite, simple, and undirected. For a graph $G=(V,E)$, we write
$n=|V|$ and $m=|E|$. The complement of $G$ is denoted by $\overline G$, and
\[
   \overline m = |E(\overline G)|=\binom n2-m .
\]
For a subset $W\subseteq V$, we denote by $G[W]$ the subgraph of $G$
induced by $W$.

For a vertex $v\in V$, let
\[
   N_G(v)=\{w\in V : vw\in E(G)\}
\]
be the open neighborhood of $v$, and let $N_G[v]=N_G(v)\cup\{v\}$ be its
closed neighborhood. We write
\[
   \overline N_G(v)=V\setminus N_G[v]
\]
for the set of non-neighbors of $v$ in $G$. Equivalently,
$\overline N_G(v)=N_{\overline G}(v)$. When the graph is clear from the
context, we omit the subscript. We denote by $\omega(G)$ the size of a
maximum clique of $G$.

A graph is \emph{bipartite} if its vertex set can be partitioned into two
independent sets, and \emph{co-bipartite} if its complement is bipartite,
or equivalently, if its vertex set can be partitioned into two cliques.
Thus, bipartite graphs are precisely the $(2,0)$-graphs and co-bipartite
graphs are precisely the $(0,2)$-graphs.

We now recall the two algorithmic ingredients on which our improvements
are based.

\subsection{The recognition algorithm for $(2,1)$-graphs}

Brandst\"adt, Le and Szymczak~\cite{BRANDSTADT199859} gave an
$O((n+m)^2)$-time algorithm for recognizing $(2,1)$-graphs. We reproduce
below the part of their algorithm that is relevant for our purposes, with
minor notational adaptations. The procedure \textsc{Modify}, used in the
case where the residual set contains a triangle, is treated as a black box;
the only fact needed here is that all calls to \textsc{Modify} made by the
algorithm, except for the final residual case, fit within the
$O(n(n+m))$ part of the running time.

\begin{algorithm}[tbp]
\footnotesize
\caption{Recognition of $(2,1)$-graphs, after Brandst\"adt, Le and Szymczak~\cite{BRANDSTADT199859}}
\label{alg:BLS}
\begin{algorithmic}[1]
\Require A graph $G=(V,E)$.
\Ensure Decide whether $G$ is a $(2,1)$-graph.
\ForAll{$v\in V$}
   \State check whether $G[N(v)]$ is a split graph
   \State check whether $G[N(v)]$ is bipartite
\EndFor
\If{there exists $v\in V$ such that $G[N(v)]$ is neither split nor bipartite}
   \State \Return no
\EndIf
\State
\[
   C_F \gets
   \{v\in V : G[N(v)] \text{ is not split and } G[N(v)] \text{ is bipartite}\}
\]
\State
\[
   B_F \gets
   \{v\in V : G[N(v)] \text{ is split and } G[N(v)] \text{ is not bipartite}\}
\]
\If{$C_F$ is not a clique or $G[B_F]$ is not bipartite}
   \State \Return no
\EndIf
\State
\[
   R \gets
   \{v\in V : G[N(v)] \text{ is split and } G[N(v)] \text{ is bipartite}\}.
\]
\If{$R=\emptyset$}
   \State \Return yes
\ElsIf{$\omega(G[R])\geq 3$}
   \State choose a triangle $\{a,b,c\}$ in $G[R]$
   \ForAll{$x\in\{a,b,c\}$}
      \State run \textsc{Modify}$(x)$
   \EndFor
   \If{all three calls fail}
      \State \Return no
   \Else
      \State \Return yes
   \EndIf
\Else
   \Comment{residual case: $\omega(G[R])\leq 2$}
   \ForAll{cliques $C_R\subseteq R$ with $0\leq |C_R|\leq 2$}
      \State $C\gets C_F\cup C_R$
      \State $B\gets B_F\cup (R\setminus C_R)$
      \If{$C$ is a clique and $G[B]$ is bipartite}
         \State \Return yes
      \EndIf
   \EndFor
   \State \Return no
\EndIf
\end{algorithmic}
\end{algorithm}

The sets $C_F$, $B_F$, and $R$ have the following interpretation. In any
valid partition of $V$ into a clique $C$ and a bipartite set $B$, the
vertices of $C_F$ are forced to belong to $C$, the vertices of $B_F$ are
forced to belong to $B$, and the vertices of $R$ remain undecided. The
algorithm is already bounded by $O(n(n+m))$ except for the final residual
case $\omega(G[R])\leq 2$, where it may inspect quadratically many cliques
of size at most two in $R$. Section~\ref{sec:21} shows how to replace this
last enumeration by a linear number of tests.

Since $G$ is a $(k,\ell)$-graph if and only if $\overline G$ is an
$(\ell,k)$-graph, an algorithm for $(2,1)$-graphs immediately gives an
algorithm for $(1,2)$-graphs by applying it to the complement graph.

\subsection{The sparse--dense framework}

We also use the sparse--dense partition framework of
Feder, Hell, Klein and Motwani~\cite{DBLP:conf/stoc/FederHKM99}. Let
$\mathcal S$ and $\mathcal D$ be two graph classes, called sparse and dense,
respectively. Assume that there is a constant $c$ such that, whenever
$S\in\mathcal S$ and $D\in\mathcal D$ are induced subgraphs of the same
graph, their vertex sets intersect in at most $c$ vertices.

A \emph{sparse--dense partition} of a graph $G=(V,E)$ with respect to
$(\mathcal S,\mathcal D)$ is a partition $(V_S,V_D)$ of $V$ such that
$G[V_S]\in\mathcal S$ and $G[V_D]\in\mathcal D$.

For a set $X\subseteq V$, define
\[
\mathcal N_{2c+1}(X)
 =
 \{Y\subseteq V : |X\setminus Y|\leq c,\ 
                  |Y\setminus X|=|X\setminus Y|+1\},
\]
and
\[
\mathcal N_{2c}(X)
 =
 \{Y\subseteq V : |X\setminus Y|\leq c,\ 
                  |Y\setminus X|\leq |X\setminus Y|\}.
\]
The first family consists of the local moves that increase the size of the
sparse side by one, while the second consists of the local moves of
nonpositive size change around a local optimum.

\begin{algorithm}[tbp]
\footnotesize
\caption{Sparse--dense local search of Feder, Hell, Klein and Motwani~\cite{DBLP:conf/stoc/FederHKM99}}
\label{alg:FHKM}
\begin{algorithmic}[1]
\Require A graph $G=(V,E)$; graph classes $\mathcal S,\mathcal D$ with
intersection bound $c$; recognition procedures for $\mathcal S$ and
$\mathcal D$.
\Ensure All sparse--dense partitions of $G$ with respect to
$(\mathcal S,\mathcal D)$.
\State $S^\star\gets\emptyset$
\Repeat
   \State
   \[
      \mathcal F\gets
      \{S'\in\mathcal N_{2c+1}(S^\star): G[S']\in\mathcal S\}.
   \]
   \If{$\mathcal F\neq\emptyset$}
      \State choose $S'\in\mathcal F$ and set $S^\star\gets S'$
   \Else
      \State stop the local search
   \EndIf
\Until{the local search has stopped}
\State $\Pi\gets\emptyset$
\ForAll{$S''\in\mathcal N_{2c}(S^\star)$}
   \If{$G[S'']\in\mathcal S$ and $G[V\setminus S'']\in\mathcal D$}
      \State add $(S'',V\setminus S'')$ to $\Pi$
   \EndIf
\EndFor
\State \Return $\Pi$
\end{algorithmic}
\end{algorithm}

The following theorem is the algorithmic consequence of the framework.

\begin{theorem}[Feder, Hell, Klein and Motwani~\cite{DBLP:conf/stoc/FederHKM99}]
\label{thm:FHKM}
Let $\mathcal S$ and $\mathcal D$ be graph classes with intersection bound
$c$. If membership in both classes can be tested in time $T(n)$, then every
$n$-vertex graph has at most $n^{2c}$ sparse--dense partitions with respect
to $(\mathcal S,\mathcal D)$, and all of them can be found in time
$O(n^{2c+2}T(n))$.
\end{theorem}

In particular, recognizing $(2,2)$-graphs fits this framework by taking
$\mathcal S$ to be the class of bipartite graphs and $\mathcal D$ to be the
class of co-bipartite graphs. The intersection bound is then $c=4$, because
a graph that is both bipartite and co-bipartite has at most four vertices in
each common sparse--dense intersection. Since bipartite and co-bipartite
graphs can be recognized in linear time, this gives an
$O(n^{10}(n+m))$-time recognition algorithm for $(2,2)$-graphs. Our
improvement in Section~\ref{sec:22} starts from this framework and reduces
the local-search enumeration by using the additional structure of such
intersections.

\section{$(2,1)$- and $(1,2)$-graphs}\label{sec:21}

In this section we improve the last branch of Algorithm~\ref{alg:BLS},
namely the residual case in which the undecided set $R$ induces no triangle.
This yields a faster recognition algorithm for $(2,1)$-graphs. The result
for $(1,2)$-graphs then follows by complementation.

We first record the structural fact that allows us to reduce the number of
candidate pairs that have to be inspected.

\begin{proposition}\label{prop:odd-cycle-edges}
Let $H$ be a non-bipartite graph on $n$ vertices, and let $Q$ be a shortest
odd cycle of $H$. Then the number of edges of $H$ with at least one endpoint
in $V(Q)$ is $O(n)$.
\end{proposition}

\begin{proof}
Write $Q=v_1v_2\cdots v_{2k+1}v_1$. If $k=1$, then $Q$ is a triangle, and
the number of edges incident to its three vertices is at most $3n-6$, hence
$O(n)$.

Assume now that $k>1$. We first note that $Q$ is chordless. Indeed, any
chord of $Q$ would split $Q$ into two cycles, one of which is odd and shorter
than $Q$, contradicting the choice of $Q$.

We show that every vertex $x\in V(H)\setminus V(Q)$ has at most two
neighbors in $Q$. Suppose, without loss of generality, that $x$ is adjacent
to $v_1$. Then $x$ cannot be adjacent to $v_2$ or to $v_{2k+1}$, since either
adjacency would create a triangle. Moreover, for each $4\leq j\leq 2k-1$,
the vertex $x$ cannot be adjacent to $v_j$: if $j$ is even, then
$xv_1v_2\cdots v_jx$ is an odd cycle shorter than $Q$, while if $j$ is odd,
then $xv_jv_{j+1}\cdots v_{2k+1}v_1x$ is an odd cycle shorter than $Q$.
Finally, $x$ cannot be adjacent to both $v_3$ and $v_{2k}$, because then
$xv_3v_4\cdots v_{2k}x$ would be an odd cycle of length $2k-1$.

Thus, once one neighbor of $x$ in $Q$ is fixed, the only possible second
neighbors are the two vertices at distance two from it along $Q$, and at
most one of them can occur. Hence every vertex outside $Q$ has at most two
neighbors in $Q$. Since $Q$ is chordless, the edges of $H$ with both
endpoints in $V(Q)$ are exactly the $|V(Q)|$ edges of the cycle. The edges
between $Q$ and $V(H)\setminus V(Q)$ are at most $2(n-|V(Q)|)$. Therefore
the total number of edges incident to $V(Q)$ is at most
\[
   |V(Q)|+2(n-|V(Q)|)=2n-|V(Q)|\leq 2n-5=O(n).
\]
\end{proof}

\begin{theorem}\label{theorem:2-1-graph}
It can be decided in $O(n(n+m))$ time whether a graph $G=(V,E)$ is a
$(2,1)$-graph.
\end{theorem}

\begin{proof}
We run Algorithm~\ref{alg:BLS} until its last branch. All previous steps,
including the calls to \textsc{Modify} in the case $\omega(G[R])\geq 3$,
fit within the $O(n(n+m))$ part of the running time. Therefore it remains
to replace the final enumeration in the case
\[
   R\neq\emptyset
   \qquad\text{and}\qquad
   \omega(G[R])\leq 2 .
\]

At this point the vertex set is partitioned into three sets $C_F$, $B_F$,
and $R$, where $C_F$ is forced to belong to the clique side, $B_F$ is forced
to belong to the bipartite side, and the vertices of $R$ are still undecided.
More explicitly, in every feasible $(2,1)$-partition $(C,B)$ extending the
partial assignment, we must have
\[
   C_F\subseteq C,
   \qquad
   B_F\subseteq B,
   \qquad
   C\setminus C_F\subseteq R,
   \qquad
   B\setminus B_F\subseteq R .
\]
Since $\omega(G[R])\leq 2$, at most two vertices of $R$ can be placed in
the clique $C$.

First test whether $G[B_F\cup R]$ is bipartite. If it is, then
\[
   C=C_F,
   \qquad
   B=B_F\cup R
\]
is a valid partition, and we accept.

Assume, from now on, that $G[B_F\cup R]$ is not bipartite. Since
$G[B_F]$ is bipartite, every odd cycle of $G[B_F\cup R]$ contains at least
one vertex of $R$. We compute a shortest odd cycle $Q$ of $G[B_F\cup R]$.
This can be done in $O(n(n+m))$ time by running a breadth-first search from
each vertex of $R$ and keeping a shortest odd cycle found.

Let $(C,B)$ be any valid partition extending $C_F$ and $B_F$. Since
$G[B]$ is bipartite, the odd cycle $Q$ cannot be entirely contained in $B$.
Thus at least one vertex of $V(Q)\cap R$ must belong to $C$. As observed
above, at most two vertices of $R$ can belong to $C$. Therefore every valid
partition is captured by one of the following two types of choices.

\begin{enumerate}
    \item A single vertex $v\in V(Q)\cap R$ is added to the clique side:
    \[
       C=C_F\cup\{v\},
       \qquad
       B=B_F\cup (R\setminus\{v\}).
    \]

    \item Two adjacent vertices $v,w\in R$ are added to the clique side,
    with at least one of them belonging to $V(Q)$:
    \[
       C=C_F\cup\{v,w\},
       \qquad
       B=B_F\cup (R\setminus\{v,w\}).
    \]
\end{enumerate}

For each candidate produced in this way, we check in $O(n+m)$ time whether
$C$ is a clique and whether $G[B]$ is bipartite. The number of candidates
of the first type is at most $|V(Q)|\leq n$.
The number of candidates of the second type is bounded by the number of
edges of $G[B_F\cup R]$ incident to $V(Q)$, even though only edges with both
endpoints in $R$ are tested. This number is $O(n)$ by
Proposition~\ref{prop:odd-cycle-edges}.
Hence only $O(n)$ candidates are tested.

Consequently, the residual case takes $O(n(n+m))$ time. This gives the
claimed total running time.
\end{proof}

\begin{remark}\label{rem:cai-schieber}
The residual case can also be handled using the linear-time algorithm of
Cai and Schieber~\cite{CAI199727} for computing the intersection of all odd
cycles of a graph. Let
\[
   R^\star=\{v\in R : C_F\subseteq N_G(v)\}
\]
be the set of vertices of $R$ that can be added to the clique side together
with $C_F$. After the initial test showing that $G[B_F\cup R]$ is not
bipartite, one may compute the set of vertices that belong to every odd
cycle of $G[B_F\cup R]$. If this set contains a vertex $v\in R^\star$, then
moving $v$ to the clique side makes the remaining graph bipartite. Otherwise,
one may try each vertex $v\in R^\star$ and apply the same procedure to
$G[B_F\cup (R\setminus\{v\})]$ in order to find a vertex
$w\in R^\star\cap N_G(v)$ that belongs to every odd cycle of this graph.
Then $\{v,w\}$ can be added to the clique side. This gives another
$O(n(n+m))$ implementation of the same residual step.
\end{remark}

\begin{corollary}\label{corollary:1-2-graph}
It can be decided in $O(n(n+\overline m))$ time whether a graph $G=(V,E)$
is a $(1,2)$-graph, where $\overline m=|E(\overline G)|$.
\end{corollary}

\begin{proof}
A graph $G$ is a $(1,2)$-graph if and only if its complement $\overline G$
is a $(2,1)$-graph. Applying Theorem~\ref{theorem:2-1-graph} to
$\overline G$ gives a running time of
\[
   O(n(n+|E(\overline G)|))=O(n(n+\overline m)).
\]
\end{proof}

\section{$(2,2)$-graphs}\label{sec:22}

In this section we refine Algorithm~2 for the recognition of
$(2,2)$-graphs. We take the sparse class to be the class of bipartite graphs
and the dense class to be the class of co-bipartite graphs. Thus, a feasible
partition is a pair $(B,C)$ such that
\[
   V(G)=B\cup C,\qquad B\cap C=\emptyset,
\]
where $G[B]$ is bipartite and $G[C]$ is co-bipartite.

The generic application of Theorem~2.1 gives an
\[
   O(n^{10}(n+m))
\]
recognition algorithm. We show how to reduce this bound to
\[
   O\bigl(n^4(n+\min\{m,\overline m\})^3\bigr),
\]
where $\overline m=|E(\overline G)|$.

The improvement has two ingredients. First, the local moves of the
sparse--dense framework can be restricted because the vertices removed from the current bipartite side lie in the intersection
of the current bipartite side and the target co-bipartite side, and therefore
induce a graph that is both bipartite and co-bipartite. Second, a preprocessing step based on the algorithms for
$(2,1)$- and $(1,2)$-graphs produces a bipartite set that is already within
constant distance of a target sparse side.

\subsection{Restricting the local moves}

In both phases of Algorithm~2, a local move is obtained by removing a set
$P$ from the current sparse side $S^\star$ and adding a set $Q$ of vertices.
We may assume that
\[
   P\cap Q=\emptyset,
\]
since any vertex that is both removed and added back can simply be ignored.

The following observation is the local-optimality property used by the
sparse--dense framework.

\begin{observation}
\label{obs:22-local-optimum}
After Phase~I of Algorithm~2, if a feasible sparse--dense partition
$(B^\star,C^\star)$ exists, then
\[
   |S^\star|\geq |B^\star|.
\]
\end{observation}

\begin{proof}
Suppose, to the contrary, that $|S^\star|<|B^\star|$. Let
\[
   P=S^\star\setminus B^\star=S^\star\cap C^\star.
\]
Since $G[S^\star]$ is bipartite and $G[C^\star]$ is co-bipartite, the set
$P$ lies in the intersection of a bipartite induced subgraph and a
co-bipartite induced subgraph. Removing $P$ from $S^\star$ leaves a subset of
$B^\star$. Since $|S^\star|<|B^\star|$, we can add $|P|+1$ vertices from
$B^\star\setminus S^\star$ and obtain a larger bipartite set. This contradicts
the fact that Phase~I has stopped.
\end{proof}

The same reasoning applies in Phase~II. If the target sparse side is
$B^\star$, then the vertices removed from the current sparse side are again
\[
   S^\star\setminus B^\star=S^\star\cap C^\star.
\]
Thus the removed set is always an intersection of a bipartite induced subgraph
and a co-bipartite induced subgraph.

\begin{lemma}
\label{lem:22-restricted-P}
In the $(2,2)$ case, in both phases of Algorithm~2, it is sufficient to
consider removed sets $P\subseteq S^\star$ such that $G[P]$ is both bipartite
and co-bipartite. In particular, $|P|\leq 4$.
\end{lemma}

\begin{proof}
Let $(B^\star,C^\star)$ be a feasible partition, where $G[B^\star]$ is
bipartite and $G[C^\star]$ is co-bipartite. The relevant removed set is
\[
   P=S^\star\setminus B^\star=S^\star\cap C^\star.
\]
Since $P\subseteq S^\star$, the graph $G[P]$ is bipartite. Since
$P\subseteq C^\star$, the graph $G[P]$ is co-bipartite. Hence $G[P]$ is both
bipartite and co-bipartite.

A graph and its complement cannot both be bipartite on more than four
vertices. Indeed, on five vertices, every bipartite graph has an independent
set of size at least three, which becomes a triangle in the complement. Thus
$|P|\leq 4$.
\end{proof}

The graphs on at most four vertices that are both bipartite and co-bipartite
are exactly
\[
   \emptyset,\quad
   K_1,\quad
   K_2,\quad
   2K_1=\overline{K_2},\quad
   K_1+K_2=\overline{K_{1,2}},\quad
   K_3-e,\quad
   2K_2=\overline{C_4},\quad
   P_4,\quad
   C_4.
\]

\begin{lemma}
\label{lem:22-generate-P}
All candidate removed sets $P$ needed by the restricted local search can be
generated in
\[
   O((n+m)^2)
\]
time.
\end{lemma}

\begin{proof}
By Lemma~\ref{lem:22-restricted-P}, it is enough to generate all vertex sets
inducing one of the nine graphs listed above.

The empty set, the one-vertex sets, and the two-vertex sets are generated
directly. This accounts for the types
\[
   \emptyset,\quad K_1,\quad K_2,\quad 2K_1.
\]
These sets are generated in $O(n^2)$ time.

For the four-vertex types, we enumerate all pairs of vertex-disjoint edges.
Let these edges be $ab$ and $cd$. We inspect the adjacencies between
$\{a,b\}$ and $\{c,d\}$. Since both $ab$ and $cd$ are edges, the induced graph
on $\{a,b,c,d\}$ is one of the relevant four-vertex types precisely when each
endpoint of one edge is adjacent to at most one endpoint of the other edge. In
that case, the induced graph is one of
\[
   2K_2,\quad P_4,\quad C_4.
\]
Thus all candidates of these three types are generated by scanning pairs of
vertex-disjoint edges. This costs $O(m^2)$ time.

For the three-vertex types with one edge, we enumerate every edge $ab$ and
every vertex $x\notin\{a,b\}$. If $x$ is adjacent to neither endpoint of
$ab$, then $G[\{a,b,x\}]$ is $K_1+K_2$. If $x$ is adjacent to exactly one
endpoint of $ab$, then $G[\{a,b,x\}]$ is $K_3-e$. If $x$ is adjacent to both
endpoints, the induced graph is a triangle and is discarded. Hence all
candidates of types
\[
   K_1+K_2,\quad K_3-e
\]
are generated in $O(nm)$ time.

Together, these procedures generate every set inducing one of the nine
possible graphs. Each generated set is checked in constant time using adjacency
queries, and duplicates do not affect the asymptotic bound. The total time is
\[
   O(n^2+nm+m^2)=O((n+m)^2).
\]
\end{proof}

Consequently, the local search can be restricted to removed sets inducing one
of these nine graphs, instead of considering all subsets of size at most four.

\subsection{Testing the added vertices}

We now describe how to test the vertices added after a removed set $P$ has
been fixed. Let
\[
   B'=S^\star\setminus P.
\]
Since $G[S^\star]$ is bipartite, the graph $G[B']$ is bipartite. We compute
the connected components of $G[B']$, together with one bipartition for each
component.

In Phase~I we seek a set
\[
   Q\subseteq V\setminus B'
\]
such that
\[
   Q\cap P=\emptyset,\qquad |Q|=|P|+1,
\]
and $G[B'\cup Q]$ is bipartite. In Phase~II the analogous search has
$|Q|\leq |P|$.

For a candidate vertex $q\notin B'\cup P$, the graph
$G[B'\cup\{q\}]$ is bipartite if and only if $q$ does not have neighbors in
both sides of the same connected component of $G[B']$. This can be tested by
scanning the neighbors of $q$.

When a vertex $q$ is successfully added, the bipartition of the new graph can
be updated as follows. Components of $G[B']$ with no neighbor of $q$ remain
unchanged. Components that have a neighbor of $q$ are merged with $q$ into one
component; their sides are oriented according to the side on which a neighbor
of $q$ lies. This update can be done in $O(n+m)$ time.

\begin{lemma}
\label{lem:22-find-Q}
For a fixed removed set $P$ with $|P|\leq 4$, the search for a suitable added
set $Q$ can be performed in
\[
   O(n^4(n+m))
\]
time. The same bound applies both to the improving moves of Phase~I and to
the non-increasing moves of Phase~II.
\end{lemma}

\begin{proof}
In Phase~I, we have $|Q|=|P|+1\leq 5$. A direct enumeration of all vertices of
$Q$ would give a factor $n^5$ in the worst case. Instead, we enumerate only the
first at most $|Q|-1$ vertices of $Q$. For each such partial choice, we maintain
the connected components and bipartitions of the current bipartite graph.

After these at most four vertices have been fixed, the last vertex, if still
needed, is found by scanning all remaining vertices and testing whether it can
be added while preserving bipartiteness. For a candidate vertex $q$, this test
is done by scanning the neighbors of $q$ and checking that $q$ does not have
neighbors in both sides of the same connected component of the current
bipartite graph. Over the scan, this costs $O(n+m)$.

Whenever a vertex is successfully added before the last step, the bipartition
of the new graph is updated by merging the components that have a neighbor of
the new vertex and orienting their sides consistently. This update is bounded
by $O(n+m)$. Since at most four vertices are explicitly enumerated and each
partial choice costs $O(n+m)$, the total time is
\[
   O(n^4(n+m)).
\]

In Phase~II, we only need $|Q|\leq |P|\leq 4$, which is no harder. The same
upper bound applies.
\end{proof}

By Lemmas~\ref{lem:22-generate-P} and~\ref{lem:22-find-Q}, one iteration of
Phase~I can be implemented in
\[
   O\bigl((n+m)^2\cdot n^4(n+m)\bigr)
   =
   O(n^4(n+m)^3)
\]
time. The same bound applies to Phase~II. If Phase~I were started from the
empty set, there could be $O(n)$ successful iterations. We next show how to
construct an initial bipartite set that reduces this number to a constant.

\subsection{Preprocessing and forced vertices}

We use the algorithms for $(2,1)$- and $(1,2)$-graphs as subroutines. For each
vertex $v\in V(G)$, we test whether
\[
   G[N(v)]\in (1,2)
\]
and whether
\[
   G[V\setminus N[v]]\in (2,1).
\]
Using the algorithms of the previous sections, the tests over all vertices take
\[
   O(n^2(n+m))
\]
time. Indeed, for a fixed vertex $v$, both induced subgraphs
$G[N(v)]$ and $G[V\setminus N[v]]$ have at most $n$ vertices and at most $m$
edges. Thus, as a uniform upper bound, the corresponding $(1,2)$- and
$(2,1)$-recognition tests cost $O(n(n+m))$ time for this vertex. Since the
tests are performed for all $v\in V(G)$, the total preprocessing cost is
$O(n^2(n+m))$.

Define
\[
\begin{aligned}
   B^F &= \{v\in V :
      G[N(v)]\in (1,2)
      \text{ and }
      G[V\setminus N[v]]\notin (2,1)\},\\
   C^F &= \{v\in V :
      G[N(v)]\notin (1,2)
      \text{ and }
      G[V\setminus N[v]]\in (2,1)\},\\
   R &= \{v\in V :
      G[N(v)]\in (1,2)
      \text{ and }
      G[V\setminus N[v]]\in (2,1)\}.
\end{aligned}
\]
If
\[
   V\neq B^F\cup C^F\cup R,
\]
then some vertex fails both tests. In that case $G$ is not a $(2,2)$-graph.
Likewise, if $G[B^F]$ is not bipartite or $G[C^F]$ is not co-bipartite, then
$G$ is not a $(2,2)$-graph.

\begin{lemma}
\label{lem:22-forced-vertices}
Let $(B^\star,C^\star)$ be a feasible $(2,2)$-partition, where
$G[B^\star]$ is bipartite and $G[C^\star]$ is co-bipartite. Then
\[
   B^F\subseteq B^\star
   \qquad\text{and}\qquad
   C^F\subseteq C^\star.
\]
\end{lemma}

\begin{proof}
If $v\in B^\star$, then the neighbors of $v$ inside $B^\star$ lie in one
independent set, while the neighbors of $v$ inside the co-bipartite side
$C^\star$ lie in two cliques. Hence $G[N(v)]$ is a $(1,2)$-graph.

If $v\in C^\star$, then the vertices in $V\setminus N[v]$ consist of a
bipartite part coming from $B^\star$ together with one clique coming from the
other clique of the co-bipartite side. Hence
$G[V\setminus N[v]]$ is a $(2,1)$-graph.

Therefore, if $G[V\setminus N[v]]$ is not a $(2,1)$-graph, the vertex $v$
cannot belong to $C^\star$, and so it is forced into $B^\star$. This proves
$B^F\subseteq B^\star$. The proof of $C^F\subseteq C^\star$ is symmetric.
\end{proof}

If $R=\emptyset$, then all vertices are forced. In this case, by
Lemma~\ref{lem:22-forced-vertices}, the graph is a $(2,2)$-graph if and only
if $G[B^F]$ is bipartite and $G[C^F]$ is co-bipartite. We may therefore assume
from now on that $R\neq\emptyset$.

Choose a vertex $v\in R$ of minimum degree. Since $v\in R$, there are
partitions
\[
   N[v]=C_1\cup C_2\cup I_3,
   \qquad
   V\setminus N[v]=C_3\cup I_1\cup I_2,
\]
where $v\in C_1$, each $C_i$ is a clique, and each $I_j$ is an independent
set, for $1\leq i,j\leq 3$.

We now seek a set
\[
   B\subseteq I_1\cup I_2\cup I_3
\]
of maximum cardinality satisfying the following four conditions:
\[
\begin{array}{ll}
{\rm (i)} & G[B]\text{ is bipartite},\\
{\rm (ii)} & B^F\cap (I_1\cup I_2\cup I_3)\subseteq B,\\
{\rm (iii)} & C^F\cap B=\emptyset,\\
{\rm (iv)} & G[(I_1\cup I_2\cup I_3)\setminus B]\text{ is co-bipartite}.
\end{array}
\]

\begin{lemma}
\label{lem:22-initial-close}
Let $(B^\star,C^\star)$ be a feasible $(2,2)$-partition. Let $B$ be a
maximum-cardinality subset of $I_1\cup I_2\cup I_3$ satisfying
conditions {\rm (i)}--{\rm (iv)} above. Then
\[
   |B|+6\geq |B^\star|.
\]
\end{lemma}

\begin{proof}
For each $j\in\{1,2,3\}$, the set
\[
   I_j\setminus B^\star=I_j\cap C^\star
\]
is an independent set contained in the co-bipartite graph $G[C^\star]$.
Hence
\[
   |I_j\setminus B^\star|\leq 2.
\]
Similarly, for each $i\in\{1,2,3\}$, the set
\[
   B^\star\cap C_i
\]
is a clique contained in the bipartite graph $G[B^\star]$, and therefore
\[
   |B^\star\cap C_i|\leq 2.
\]

Now consider
\[
   B_0=(I_1\cup I_2\cup I_3)\cap B^\star.
\]
The set $B_0$ satisfies conditions (i)--(iv). Indeed, $G[B_0]$ is bipartite
as an induced subgraph of $G[B^\star]$; by Lemma~\ref{lem:22-forced-vertices},
it contains all forced vertices of $B^F$ that lie in
$I_1\cup I_2\cup I_3$ and contains no vertex of $C^F$; finally,
\[
   (I_1\cup I_2\cup I_3)\setminus B_0
      \subseteq C^\star,
\]
and hence it induces a co-bipartite graph.

Since $B$ is chosen with maximum cardinality among all sets satisfying
(i)--(iv),
\[
   |B|\geq |B_0|
      =
   |(I_1\cup I_2\cup I_3)\cap B^\star|.
\]
Moreover,
\[
\begin{aligned}
   |B^\star\setminus (I_1\cup I_2\cup I_3)|
      &\leq |B^\star\cap C_1|
          + |B^\star\cap C_2|
          + |B^\star\cap C_3|\\
      &\leq 6.
\end{aligned}
\]
Therefore
\[
   |B|+6\geq |B^\star|.
\]
\end{proof}

\subsection{Computing the initial set}

It remains to compute the maximum set $B$ satisfying conditions (i)--(iv).
Instead of searching for $B$ directly, we enumerate the vertices excluded from
it.

Let
\[
   P=(I_1\cup I_2)\setminus B.
\]
Then $G[P]$ must be one of the nine graphs on at most four vertices listed
above. Indeed, $P\subseteq I_1\cup I_2$, and $I_1\cup I_2$ is bipartite.
Moreover, by condition {\rm (iv)}, $P$ is contained in a co-bipartite induced
subgraph. Hence $G[P]$ is both bipartite and co-bipartite. We also require
\[
   C^F\cap (I_1\cup I_2)\subseteq P.
\]
By Lemma~\ref{lem:22-generate-P}, all possible sets $P$ can be enumerated in
\[
   O((n+m)^2)
\]
time.

For each such set $P$, we choose a set
\[
   P'\subseteq I_3\setminus B^F
\]
such that
\[
   C^F\cap I_3\subseteq P',
   \qquad
   |C^F\cap I_3|\leq |P'|\leq 2.
\]
If $|C^F\cap I_3|\geq 1$, then there are $O(n)$ possible choices for $P'$.
If $C^F\cap I_3=\emptyset$, then up to two vertices in $I_3\cap R$ can be
selected to form $P'$. Since $I_3\subseteq N(v)$, we have
\[
   |I_3\cap R|\leq d(v).
\]
Because $v$ was chosen with minimum degree in $R$,
\[
   d(v)(|R|-1)
      \leq
      \sum_{u\in R\setminus\{v\}} d(u)
      \leq 2m.
\]
Therefore the number of possible choices for $P'$ is at most
\[
   1+|I_3\cap R|+\binom{|I_3\cap R|}{2}
      \leq
   1+d(v)+d(v)(|R|-1)
      \leq
   1+2m
      =
   O(m).
\]

For each pair $(P,P')$, define
\[
   B=(I_1\cup I_2\cup I_3)\setminus (P\cup P').
\]
We test in $O(n+m)$ time whether $B$ satisfies conditions (i)--(iv), and keep
a valid set of maximum cardinality.

\begin{lemma}
\label{lem:22-compute-B}
The maximum-cardinality set $B$ satisfying conditions {\rm (i)}--{\rm (iv)}
can be computed in
\[
   O((n+m)^4)
\]
time.
\end{lemma}

\begin{proof}
There are $O((n+m)^2)$ possible sets $P$. For each $P$, there are at most
$O(m)$ possible sets $P'$. Each resulting candidate $B$ is tested in
$O(n+m)$ time. Thus the total time is
\[
   O((n+m)^2\cdot m\cdot (n+m))
      =
   O(m(n+m)^3)
      \subseteq
   O((n+m)^4).
\]
\end{proof}

If no valid set $B$ exists, then $G$ is not a $(2,2)$-graph, because every
feasible partition would produce a valid set by the proof of
Lemma~\ref{lem:22-initial-close}.

\subsection{The final algorithm}

We now combine the previous ingredients. First, we perform the preprocessing
over all vertices and compute the sets $B^F$, $C^F$, and $R$. If some vertex
fails both tests, or if $G[B^F]$ is not bipartite, or if $G[C^F]$ is not
co-bipartite, then we reject.

If $R=\emptyset$, then all vertices are forced. In this case we accept if and
only if $G[B^F]$ is bipartite and $G[C^F]$ is co-bipartite.

Assume now that $R\neq\emptyset$. We choose a minimum-degree vertex
$v\in R$, compute the decompositions
\[
   N[v]=C_1\cup C_2\cup I_3,
   \qquad
   V\setminus N[v]=C_3\cup I_1\cup I_2,
\]
and compute the maximum valid initial set $B$ of
Lemma~\ref{lem:22-compute-B}. If no such set exists, then we reject.

Otherwise, we start Algorithm~2 with
\[
   S^\star:=B,
\]
and restrict all local moves to the removed sets described in
Lemma~\ref{lem:22-restricted-P}. By Lemma~\ref{lem:22-initial-close}, if a
feasible partition $(B^\star,C^\star)$ exists, then
\[
   |B|+6\geq |B^\star|.
\]

We run Phase~I until either no restricted improving move exists or six
successful improvements have been performed. If Phase~I stops because no
restricted improving move exists, then the standard sparse--dense
local-optimality argument, as stated in Observation~\ref{obs:22-local-optimum},
implies that
\[
   |S^\star|\geq |B^\star|
\]
for every feasible partition $(B^\star,C^\star)$. Otherwise, six successful
improvements have been performed. Since each successful Phase~I move increases
$|S^\star|$ by one, Lemma~\ref{lem:22-initial-close} again gives
\[
   |S^\star|\geq |B^\star|.
\]
Thus, in either case, when Phase~I terminates the current bipartite set is at
least as large as the sparse side of any feasible $(2,2)$-partition. Phase~II
is then applied using the same restricted family of removed sets.

To see why Phase~II is sufficient, let $(B^\star,C^\star)$ be a feasible
partition. For this target sparse side, the removed set is
\[
   P=S^\star\setminus B^\star,
\]
and the added set is
\[
   Q=B^\star\setminus S^\star.
\]
By Lemma~\ref{lem:22-restricted-P}, the set $P$ belongs to the restricted
family of removed sets. Since $|S^\star|\geq |B^\star|$, we have
\[
   |Q|\leq |P|.
\]
Therefore the target exchange from $S^\star$ to $B^\star$ is among the
non-increasing moves considered by Phase~II.

The running time is as follows.
\begin{itemize}
   \item The preprocessing over all vertices costs $O(n^2(n+m))$.
   \item The computation of the initial set $B$ costs $O((n+m)^4)$.
   \item Each restricted Phase~I iteration costs $O(n^4(n+m)^3)$, and the
algorithm performs at most six successful iterations before entering Phase~II.
   \item Phase~II also costs $O(n^4(n+m)^3)$.
\end{itemize}
Therefore the total running time is
\[
   O(n^4(n+m)^3).
\]

Finally, the class of $(2,2)$-graphs is closed under complementation. Thus,
if $\overline m<m$, we apply the same algorithm to $\overline G$ instead of
$G$. The construction of $\overline G$ costs $O(n^2)$ time, which is dominated
by the final bound.

\begin{theorem}
\label{thm:22-recognition}
It can be determined in
\[
   O\bigl(n^4(n+\min\{m,\overline m\})^3\bigr)
\]
time whether a graph $G=(V,E)$ belongs to the class $(2,2)$.
\end{theorem}

\section{Conclusions and open questions}
\label{sec:conclusions}

We revisited the recognition of the first nontrivial polynomial
$(k,\ell)$-graph classes beyond split, bipartite, and co-bipartite graphs.
For $(2,1)$-graphs, we showed that the residual case
$\omega(G[R])\leq 2$ in the algorithm of Brandst\"adt, Le and Szymczak,
which was responsible for the extra quadratic enumeration beyond the
$O(n(n+m))$ part of the procedure, can be handled by a linear number of
tests. By complementation, the same improvement applies to $(1,2)$-graphs.
For $(2,2)$-graphs, we refined the sparse--dense framework of Feder, Hell,
Klein and Motwani by restricting the local moves and by using a preprocessing
step based on the improved recognition algorithms for $(2,1)$- and
$(1,2)$-graphs.

Several natural questions remain open. First, although the running time for
$(2,1)$-graphs is improved to $O(n(n+m))$, it is not clear whether this is
best possible. Is there a linear-time or near-linear-time recognition
algorithm for $(2,1)$-graphs? Equivalently, can the initial neighborhood
tests, the calls to the modification procedure, and the residual case be
combined or replaced by a more direct structural algorithm?

The most significant gap concerns $(2,2)$-graphs. The algorithm presented
here improves the previously known bound from $O(n^{10}(n+m))$ to
\[
   O\bigl(n^4(n+\min\{m,\overline m\})^3\bigr),
\]
but this is still a high-degree polynomial. It would be interesting to know
whether $(2,2)$-graphs admit a substantially faster recognition algorithm,
for instance with running time $O(n^c)$ for a smaller constant $c$, or with
a bound closer to the running times known for $(2,1)$- and $(1,2)$-graphs.

A related question is whether the sparse--dense local-search framework is
inherently too general for the case of $(2,2)$-graphs. The proof of our
bound still uses local moves around a sparse side, even after the
preprocessing step. A more direct structural description of the residual
instances might lead to a simpler and faster algorithm.

More broadly, one of the refinements used here is not specific to
$(2,2)$-graphs: in any application of the sparse--dense framework, the set
removed from the current sparse side during a local move is contained in the
intersection of the current sparse side and the target dense side. Thus, beyond
its cardinality bound, the removed set has a structural origin that may further
restrict the relevant local moves. In the present paper this observation is
specialized to the bipartite/co-bipartite setting, where both classes are
hereditary and the removed set induces a graph that is simultaneously bipartite
and co-bipartite. This yields a small list of possible removed sets. It would
be interesting to explore whether the same restricted local-search viewpoint
leads to faster algorithms for other sparse--dense partition problems.

\section*{Acknowledgments}

Partially supported by CONICET (PIP 11220200100084CO), UBACyT (20020220300079BA and 20020190100126BA), and ANPCyT (PICT-2021-I-A-00755). We would like to thank Eric Brandwein for his valuable comments.

\end{document}